\documentclass[12pt,letterpaper]{article}

\usepackage[margin=1.25in]{geometry} \usepackage{enumitem} \usepackage{setspace} 

\usepackage[utf8]{inputenc} \usepackage[T1]{fontenc} \usepackage{amsmath, amssymb, amsthm} \DeclareMathOperator*{\argmax}{arg\,max} \usepackage{graphicx} \usepackage{booktabs} \usepackage{tabularx} \usepackage{multirow} \usepackage{ifthen} \usepackage{natbib} \usepackage{hyperref} \usepackage{capt-of} \usepackage{pgfplots} \usepackage{tikz} \usepackage{multicol, latexsym} \usepackage{subcaption} \pgfplotsset{compat=1.18} \usetikzlibrary{arrows.meta, decorations.pathreplacing, patterns} \usepackage{floatrow}

\hypersetup{ colorlinks=true, linkcolor=black, citecolor=black, urlcolor=black }

  \newcommand{\BR}{\mathrm{BR}} 

\newtheorem{theorem}{Theorem} \newtheorem{corollary}{Corollary} \newtheorem{lemma}{Lemma} \newtheorem{proposition}{Proposition} \newtheorem{definition}{Definition}  \newtheorem{assumption}{Assumption} \newtheorem{remark}{Remark}

\begin{document}

\title{Moral Hazard in Delegated Bayesian Persuasion}
\author{Wilfried Youmbi Fotso%
\thanks{Corresponding author. Department of Economics, Wilfrid Laurier
University (E-mail: wyoumbifotso@wlu.ca).}%
\qquad Xun Chen%
\thanks{Department of Economics, University of Western Ontario
(E-mail: xche225@uwo.ca).}%
\qquad Kun Zhu%
\thanks{Department of Economics, University of Western Ontario (kzhu89@uwo.ca)}} \date{\today}
\maketitle

\begin{abstract}
\singlespacing
This paper studies Bayesian persuasion when information acquisition is delegated to an intermediary whose experiment choice is privately selected and costly. The principal cannot choose an information structure directly but instead designs action-contingent transfers to induce the mediator's information acquisition decision. We derive geometric sufficient conditions under which the principal's preferred experiment can be implemented despite hidden information acquisition. We then characterize the delegated second-best through an implementation-cost representation that summarizes the agency problem by the minimum expected transfer required to induce each implementable experiment. Under unrestricted transfers, this cost admits an exact closed-form expression, reducing the second-best problem to joint-surplus maximization over the implementable set. The framework identifies two sources of distortion generated by delegation: implementability constraints and endogenous implementation costs. Under entropy costs, we derive tangency conditions for implementable interior two-posterior maximizers of the unrestricted joint-surplus problem, together with a Shannon odds-ratio relation.
\end{abstract}


\bigskip

\section{Introduction}

Bayesian persuasion studies how a principal strategically designs information structures to influence the actions of a receiver. Since the seminal work of \citet{Kamenica2011}, the literature has largely assumed that the principal directly controls the information structure and can therefore commit to any Bayes-plausible experiment. In many economically important environments, however, this assumption is unrealistic. Information is often produced by an intermediary whose investigation process is privately chosen, costly, and imperfectly observable. Credit rating agencies determine how intensively to evaluate firms before issuing ratings, consultants privately decide how much market research to conduct before making recommendations, and AI auditors choose how extensively to test algorithmic systems before certifying safety or regulatory compliance. In these settings, the principal cannot contract directly on the intermediary's information acquisition decision, but only on observable outcomes generated after information has been collected.

This paper studies Bayesian persuasion under delegated information acquisition. A principal wishes to influence a receiver through an intermediary who privately chooses an information structure from the set of Bayes-plausible experiments while incurring a convex cost of information acquisition. The intermediary may also possess intrinsic preferences over the receiver's induced actions, creating a conflict between the principal's preferred experiment and the mediator's privately optimal information policy. The central questions are: when can the principal decentralize the first-best experiment despite hidden information acquisition, and how does delegation distort information provision when first-best implementation is impossible?

The key difficulty is that the principal cannot directly optimize over the space of information structures. Instead, she chooses an action-contingent transfer rule that induces the intermediary's experiment through incentives. Delegated persuasion therefore differs from the standard persuasion problem in two fundamental respects. First, not every Bayes-plausible experiment can necessarily be implemented through observable transfers. Second, even among implementable experiments, different information structures require different transfer expenditures to induce. Delegation thus introduces both an implementability constraint and an endogenous implementation cost.

The paper develops a unified framework for analyzing these two frictions. We first derive geometric sufficient conditions under which the principal's preferred experiment can be implemented despite hidden information acquisition. The analysis shows that implementation depends on whether the desired experiment can be supported through an action-contingent transfer together with a supporting hyperplane of the information-cost functional. A simple sufficient condition is obtained under global alignment of principal and mediator preferences, where a flat transfer decentralizes the first-best allocation.

The paper then develops an implementation-cost representation of the second-best problem. Rather than solving the original bilevel contracting problem directly, we summarize the agency friction by the minimum expected transfer required to induce each implementable experiment. This implementation-cost functional transforms the principal's problem into an optimization over implementable experiments. Under unrestricted transfers, it admits an exact closed-form expression, reducing the delegated persuasion problem to joint-surplus maximization over the implementable set. This representation isolates the precise channels through which delegation alters information design: a feasibility distortion arising from limited implementability and a cost distortion arising from the transfer payments needed to induce information acquisition.

The framework also yields comparative statics on the role of the mediator's intrinsic preferences. Convex intrinsic preferences increase the value of posterior dispersion in the joint-surplus objective, whereas concave preferences reduce it. Bayes-linear preferences contribute the same value to every Bayes-plausible experiment and therefore leave the ranking of experiments unchanged whenever the principal's preferred experiment remains implementable. These results characterize how delegation changes the principal's incentives to acquire information without requiring a complete ordering of experiments under the Blackwell relation.

Under entropy-based information costs, a binary-state environment admits explicit tangency conditions and a Shannon odds-ratio relation for interior two-posterior maximizers of the unrestricted joint-surplus problem. If such an experiment is implementable, it is second-best optimal. The entropy specification also yields a direct information-theoretic interpretation of implementation costs: under Shannon entropy, the information-cost component of the minimum expected transfer coincides exactly with the mutual information generated by the induced experiment.

This paper contributes to the literature on Bayesian persuasion by introducing hidden information acquisition into information design. More broadly, it connects Bayesian persuasion, delegated information acquisition, and incentive theory through a common geometric framework. By identifying implementation costs as the key object linking contracting and persuasion, the paper provides a tractable foundation for studying delegated information production in applications such as financial certification, AI auditing, consulting, platform governance, and regulatory oversight.

The remainder of the paper is organized as follows. Section~\ref{sec:lit} reviews the related literature. Section~\ref{sec:model_delegated} presents the model. Section~\ref{sec:fb_benchmarks} studies first-best implementation. Section~\ref{subsec:virtual_persuasion} characterizes the second-best problem. Section~\ref{subsec:entropy_geometry} specializes the analysis to entropy costs. Section~\ref{sec:conclusion} concludes.

\section{Related Literature}
\label{sec:lit}

This paper contributes to the literature on Bayesian persuasion initiated by \citet{Kamenica2011}, which shows that a sender can influence a receiver's action by committing ex ante to an information structure. Subsequent work has extended this framework to environments with multiple receivers, strategic interaction, and constraints on information design \citep{BergemannMorris2016}.

The current paper is closely related to recent studies that introduce incentive problems into persuasion environments \citep[see, e.g.,][]{whitmeyer2022buying,Bharadwaj2025contracting, boleslavsky2018bayesian,Yoder2022,Deb2023}. \citet{Bharadwaj2025contracting} study contracting an intermediary who can affect a receiver's binary actions via information disclosure. \citet{boleslavsky2018bayesian} analyze a setting in which the sender designs the signal directly but must incentivize an agent whose costly effort affects the state. \citet{Yoder2022} studies delegation of information acquisition when a principal contracts with a researcher who privately knows investigation costs. \citet{Deb2023} analyze a joint screening-and-persuasion problem. In contrast to the screening problems in \citet{Yoder2022} and \citet{Deb2023}, we abstract from private information about the intermediary's type and focus instead on moral hazard in information design.

A related strand studies persuasion when information acquisition or disclosure is costly. \citet{GentzkowKamenica2014} introduce costly Bayesian persuasion. \citet{DworczakMartini2019} develop a price-theoretic representation of optimal persuasion in which supporting prices verify the optimality of information structures. \citet{Kolotilin2018} studies optimal persuasion with a privately informed receiver. The literature on rational inattention \citep[see, e.g.,][]{Sims2003,CaplinDean2013,MatejkaMcKay2015, CaplinDeanMartin2022} provides micro-foundations for the entropy-based information costs used in our analysis.

\paragraph{Direct comparison with \citet{whitmeyer2022buying}.}
Whitmeyer and Zhang study costly learning under contracts contingent on reports and the realized state, and establish efficient implementation results with a risk-neutral agent and unrestricted transfers. In our model, transfers depend only on receiver actions, so implementability can remain restrictive even when transfers are unrestricted.

\paragraph{Comparison with \citet{Yoder2022} and
\citet{boleslavsky2018bayesian}.} Yoder's researcher chooses an experiment and has privately known acquisition costs. In \citet{boleslavsky2018bayesian}, the sender retains direct control of the signal and the hidden action instead shifts the underlying state distribution, so the moral hazard problem sits upstream of persuasion rather than in the choice of experiment itself. In our environment, by contrast, the intermediary's hidden choice \emph{is} the full Bayes-plausible experiment \(\tau \in \mathcal K\), subject to a general convex cost functional on the space of experiments, and the principal contracts only on the induced action.

\section{Model}
\label{sec:model_delegated}

This section develops a model of delegated Bayesian persuasion with hidden information acquisition. The principal wishes to influence a receiver's action through information disclosure but cannot choose the information structure directly. Instead, information is produced by a mediator who privately selects an experiment and bears the associated cost of information acquisition. The principal can contract only on observable receiver actions, not on the mediator's information production decision. This informational agency problem is the central friction studied throughout the paper.

\subsection{Environment}
\label{subsec:setup}

There are three players: a principal, a mediator, and a receiver.

The state space \(\Omega=\{\omega_1,\ldots,\omega_n\}\) is finite, with common prior \(\mu_0\in\Delta(\Omega)\).

After observing information, the receiver chooses an action \(a\in A\), where the action set \(A\) is finite.

Following \citet{Kamenica2011}, experiments can be represented, without loss of generality, by Bayes-plausible distributions over posterior beliefs. Accordingly, an experiment is identified with a probability measure \(\tau\in\Delta(\Delta(\Omega))\) satisfying Bayes plausibility, \(\int_{\Delta(\Omega)} \mu\, d\tau(\mu) = \mu_0\).

Let \(\mathcal K := \left\{ \tau\in\Delta(\Delta(\Omega)) : \int\mu\,d\tau(\mu)=\mu_0 \right\}\) denote the set of all Bayes-plausible experiments.

For each posterior belief \(\mu\), the receiver chooses an optimal action

\[
a^*(\mu) \in \arg\max_{a\in A} \mathbb E_\mu \bigl[ u(a,\omega) \bigr].
\]

Throughout the paper we fix an arbitrary measurable selection whenever the receiver is indifferent.

\begin{assumption}[Receiver regularity]
\label{ass:receiver}

For every \(a\in A\), \(\mu \longmapsto \mathbb E_\mu [u(a,\omega)]\) is continuous, and a measurable selection \(a^*(\mu)\) is fixed whenever multiple optimal actions exist.

\end{assumption}

The measurable-selection assumption ensures that posterior beliefs determine receiver actions in a measurable way, allowing every payoff to be represented as a function of posterior beliefs alone.

\paragraph{Information acquisition.}

The mediator privately chooses an experiment \(\tau\in\mathcal K\), and incurs information-acquisition cost \(c(\tau)\), where \(c: \mathcal K \rightarrow \mathbb R_+\) is convex, lower semicontinuous, and normalized by \(c(\delta_{\mu_0})=0\).

Convexity means that mixing experiments costs no more than the corresponding mixture of their costs.

A leading example is the posterior-separable entropy cost \(c(\tau) = H(\mu_0) - \int H(\mu) \,d\tau(\mu)\), where \(H\) is a continuous, strictly concave uncertainty index. Under Shannon entropy, \(c(\tau)\) coincides with the mutual information generated by the experiment.

\paragraph{Transfers.}

The principal cannot observe the mediator's experiment. Instead, she commits ex ante to an action-contingent transfer rule \(t: A \rightarrow \mathbb R\), where \(t(a)\) is paid whenever the receiver ultimately chooses action \(a\).

Thus transfers depend only on observable receiver actions. They cannot depend directly on posterior beliefs or on the mediator's chosen experiment.

This restriction is the central contractual friction in the model. Although the mediator chooses an entire Bayes-plausible experiment, the principal controls incentives only through the finite-dimensional vector \(\bigl( t(a) \bigr)_{a\in A}\).

Consequently, the principal cannot implement an arbitrary perturbation of posterior payoffs across the belief simplex.

\paragraph{Payoffs.}

Let \(\pi_P(a,\omega)\) denote the principal's payoff when action \(a\) is taken in state \(\omega\).

Given posterior belief \(\mu\),

the principal's expected gross payoff is \(V_P^0(\mu) := \mathbb E_\mu \left[ \pi_P(a^*(\mu),\omega) \right]\).

Taking transfers into account, the principal's payoff becomes \(V_P(\mu;t) = V_P^0(\mu) - t(a^*(\mu))\).

The mediator may have intrinsic preferences over induced actions. Let \(\tilde v(a,\omega)\) denote the mediator's intrinsic utility.

Its expected intrinsic payoff is \(V_M^0(\mu) := \mathbb E_\mu \left[ \tilde v(a^*(\mu),\omega) \right]\).

Throughout the paper the mediator is assumed to have quasi-linear preferences.

\begin{assumption}[Quasi-linearity]
\label{ass:mediator_ql}

The mediator's payoff equals \(V_M(\mu;t) = V_M^0(\mu) + t(a^*(\mu))\).
\end{assumption}

Adding \(V_P(\mu;t) = V_P^0(\mu) - t(a^*(\mu))\) and \(V_M(\mu;t) = V_M^0(\mu) + t(a^*(\mu))\) yields \(V_P(\mu;t) + V_M(\mu;t) = V_P^0(\mu) + V_M^0(\mu)\), which is independent of transfers.

Transfers therefore redistribute surplus without changing total surplus conditional on the induced posterior. Their role is purely incentive-theoretic.

The decomposition \(V_P(\mu;t) + V_M(\mu;t) = V_P^0(\mu) + V_M^0(\mu)\) also identifies the source of delegated-information distortions.

The principal evaluates experiments according to \(V_P^0\), whereas the mediator evaluates them according to \(V_M^0\).

Whenever these posterior payoff indices differ, the mediator's privately optimal experiment need not coincide with the principal's preferred experiment.

Because receiver actions are determined by posterior beliefs, both payoff indices are functions defined directly on the belief simplex.

Consequently, the delegated persuasion problem can be analyzed entirely in belief space through the pair \(\left( V_P^0, V_M^0 \right)\), the information-cost functional \(c\), and the Bayes-plausibility constraint \(\int_{\Delta(\Omega)} \mu\, d\tau(\mu) = \mu_0\).

A function \(f: \Delta(\Omega) \rightarrow \mathbb R\) is said to be \emph{Bayes-linear} if \(f(\mu) = a + b\cdot\mu\) for some constants \(a\) and \(b\).

Consequently, \(\int f(\mu) \,d\tau(\mu) = f(\mu_0) \qquad \forall \tau \in \mathcal K\).

Bayes-linear functions therefore contribute the same expected value under every Bayes-plausible experiment. They play an important role in the characterization of delegated persuasion because they do not affect the ranking of experiments.

\subsection{Timing and Equilibrium}
\label{subsec:timing}

The timing of the delegated persuasion game is as follows.

\begin{enumerate}

\item
The principal commits publicly to an action-contingent transfer rule \(t:A\rightarrow\mathbb R\).
\item
The mediator decides whether to accept the contract. If the contract is rejected, the mediator receives its outside option \(\bar U\), no information is produced, and the game ends.

\item
Conditional on acceptance, the mediator privately chooses an experiment \(\tau\in\mathcal K\), and incurs information-acquisition cost \(c(\tau)\).
\item
A posterior belief \(\mu\) is realized according to the chosen experiment.

\item
Every feasible experiment automatically discloses its calibrated posterior to the receiver; the mediator cannot alter or suppress the report. The receiver observes \(\mu\) and chooses \(a^*(\mu)\).
\item
Transfers are paid and payoffs are realized.

\end{enumerate}

The mediator's acceptance decision is made before information is acquired. Consequently, participation is an \emph{ex ante} constraint on the mediator's optimized continuation payoff rather than a restriction on each candidate experiment individually.

For any transfer rule \(t\) with \(\BR(t)\ne\varnothing\), the mediator solves

\begin{equation}
\label{eq:mediator_problem} \max_{\tau\in\mathcal K} \left\{ \int V_M(\mu;t) \,d\tau(\mu) - c(\tau) \right\}.
\end{equation}

The mediator accepts the contract if and only if the value of \eqref{eq:mediator_problem} satisfies

\begin{equation}
\label{eq:participation} \max_{\tau\in\mathcal K} \left\{ \int V_M(\mu;t) \,d\tau(\mu) - c(\tau) \right\} \ge \bar U.
\end{equation}

The participation constraint therefore applies only to the optimized continuation payoff. Once a contract has been accepted, the mediator's choice among experiments is governed solely by the optimization problem \eqref{eq:mediator_problem}.

The mediator's best-response correspondence is

\begin{equation}
\label{eq:BR} \BR(t) := \argmax_{\tau\in\mathcal K} \left\{ \int V_M(\mu;t) \,d\tau(\mu) - c(\tau) \right\}.
\end{equation}

The principal anticipates this response when choosing transfers. Accordingly, the delegated persuasion problem is

\begin{equation}
\label{eq:principal_problem} \sup_{t:A\rightarrow\mathbb R,\,\tau\in\mathcal K} \; \int V_P(\mu;t) \,d\tau(\mu)
\end{equation}

subject to \(\tau \in \BR(t)\), and

\[
\max_{\tau'\in\mathcal K} \left\{ \int V_M(\mu;t) \,d\tau'(\mu) - c(\tau') \right\} \ge \bar U.
\]

Equation~\eqref{eq:principal_problem} is a bilevel optimization problem. The principal does not choose an experiment directly. Instead, she chooses a transfer rule that induces the mediator's private optimization problem, and the resulting experiment emerges as an equilibrium outcome.

Throughout the paper we adopt \emph{weak implementation}: any mediator best response satisfying participation can be selected. The principal optimizes over such contract--best-response pairs; contracts with no best response are infeasible.

The principal's optimization problem highlights the central difficulty created by delegation. Unlike the standard Bayesian persuasion benchmark, the principal cannot optimize directly over the set of Bayes-plausible experiments. Instead, every feasible experiment must arise as the optimal response to some action-contingent transfer rule.

This distinction motivates the analysis developed in the next section. Rather than attempting to characterize the mediator's best-response correspondence directly, we characterize the subset of experiments that are implementable under some transfer rule and summarize the agency friction through an endogenous implementation-cost functional.

\subsection{Regularity Assumption}
\label{subsec:regularity}

The analysis relies on a standard regularity condition on the information-cost functional.

\begin{assumption}[Representable subgradient]
\label{ass:subgradient}

For every experiment \(\tau\in\mathcal K\) considered in the analysis, there exists a continuous function \(\psi : \Delta(\Omega) \rightarrow \mathbb R\) such that \(c(\tau') \ge c(\tau) + \int \psi(\mu) \,d(\tau'-\tau)(\mu) \qquad \forall \tau' \in \mathcal K\).

Equivalently, \(\psi \in \partial c(\tau)\).
\end{assumption}

Assumption~\ref{ass:subgradient} ensures that the information-cost functional admits a supporting hyperplane at every experiment used in the analysis. The function \(\psi\) measures the marginal cost of locally perturbing the distribution of posterior beliefs.

For posterior-separable information costs with a continuous posterior index, this assumption is automatically satisfied. For example, under entropy costs, \(c(\tau) = H(\mu_0) - \int H(\mu) \,d\tau(\mu)\), the functional is affine in \(\tau\), and an exact subgradient is simply \(\psi(\mu) = - H(\mu)\), up to the addition of any Bayes-linear function. More generally, if \(c(\tau) = K + \int \phi(\mu) \,d\tau(\mu)\), with continuous \(\phi\), then \(\psi(\mu)=\phi(\mu)\) is an exact subgradient at every experiment.

The subgradient assumption is introduced for the geometric implementation results. Rather than carrying the functional-analytic machinery throughout the paper, we invoke Assumption~\ref{ass:subgradient} whenever supporting-hyperplane arguments are required.

\section{First-Best Benchmarks and Implementability}
\label{sec:fb_benchmarks}

We begin by characterizing the benchmark allocation the principal would choose under direct control of information acquisition. This benchmark isolates the informational component of the problem from the incentive distortions generated by delegation. The section then asks whether this benchmark can be decentralized when the mediator privately chooses the information structure.

The key tension is immediate. The principal evaluates experiments according to the gross payoff index \(V_P^0\), whereas the mediator evaluates them according to its own payoff index \(V_M^0\), transfers, and information costs. Delegation is therefore innocuous only when the principal can design a transfer rule that induces the mediator to select the principal's preferred experiment.

\begin{definition}[First-best experiment]
\label{def:fb_direct}
A first-best experiment is any \(\tau^{FB}\in\mathcal K\) satisfying
\begin{equation}
\label{eq:fb_direct} \tau^{FB} \in \argmax_{\tau\in\mathcal K} \left\{ \int V_P^0(\mu)\,d\tau(\mu)-c(\tau) \right\}.
\end{equation}
\end{definition}

This benchmark is the standard costly persuasion problem in which the principal acquires information directly and bears its cost. The central question is whether there exists a transfer rule that makes \(\tau^{FB}\) a best response of the mediator.

\subsection{A Simple Sufficient Condition}
\label{subsec:fb_impl}
We first present a transparent sufficient condition for first-best implementation. The condition requires global alignment between the principal's and mediator's posterior payoff indices. A flat transfer then shifts the mediator's continuation utility without affecting its ranking of experiments.

\begin{theorem}[Global affine alignment]
\label{thm:fb_suff_support_alignment}
Maintain Assumptions~\ref{ass:receiver}--\ref{ass:mediator_ql}. Let \(\tau^{FB}\) satisfy~\eqref{eq:fb_direct}. Suppose there exists \(\beta\in\mathbb R\) such that \(V_M^0(\mu) = V_P^0(\mu)+\beta \qquad \forall\mu\in\Delta(\Omega)\).

Then there exists a constant \(\bar t\in\mathbb R\) such that the flat transfer \(t(a)=\bar t, \qquad \forall a\in A\), weakly implements \(\tau^{FB}\) and satisfies participation.
\end{theorem}

\begin{proof}
See Appendix.
\end{proof}

The intuition is immediate. Under a flat transfer the mediator's objective differs from the principal's first-best objective only by an additive constant. Consequently, both players rank every Bayes-plausible experiment identically, so the principal's preferred experiment is also optimal for the mediator. The constant transfer can then be adjusted upward, if necessary, to satisfy participation without changing the mediator's ranking of experiments.

\begin{remark}[Why proportional alignment is insufficient]
\label{rem:alpha_one}
Suppose instead that \(V_M^0(\mu) = \alpha V_P^0(\mu)+\beta\), for some \(\alpha>0\). Under a flat transfer the mediator solves
\[
\max_{\tau\in\mathcal K} \left\{ \int \alpha V_P^0(\mu)\,d\tau(\mu) - c(\tau) \right\},
\]
which is equivalent to
\[
\max_{\tau\in\mathcal K} \left\{ \int V_P^0(\mu)\,d\tau(\mu) - \frac{1}{\alpha}c(\tau) \right\}.
\]
Unless \(\alpha=1\), the mediator places a different effective weight on information costs than the principal. Hence proportional alignment is generally insufficient to implement the first-best through a flat transfer. This observation is specific to flat transfers; richer action-contingent transfers may offset some forms of proportional misalignment.
\end{remark}

Theorem~\ref{thm:fb_suff_support_alignment} deliberately imposes a strong global alignment condition. The next subsection develops a more general geometric sufficient condition allowing arbitrary action-contingent transfers.

\subsection{A General Supporting Condition}
\label{subsec:fb_general}

The previous theorem relies on global payoff alignment. More generally, first-best implementation depends on whether one can construct an action-contingent transfer that makes the desired experiment optimal for the mediator after accounting for the information cost.

The following proposition provides a sufficient condition based on the supporting geometry of the information-cost functional.

\begin{proposition}[Sufficient supporting condition]
\label{prop:full_fb_implementability}
Let \(\tau^{FB}\in\mathcal K\) denote a first-best experiment. Suppose there exist a transfer rule \(t:A\rightarrow\mathbb R\) and a representable subgradient \(\psi^{FB}\in\partial c(\tau^{FB})\) such that \(\int V_M(\mu;t)\, d\tau^{FB}(\mu) - c(\tau^{FB}) \ge \bar U\), and, for every \(\tau\in\mathcal K\),

\begin{equation}
\label{eq:fb_support} \int \left[ V_M(\mu;t)-\psi^{FB}(\mu) \right] d\tau(\mu) \le \int \left[ V_M(\mu;t)-\psi^{FB}(\mu) \right] d\tau^{FB}(\mu).
\end{equation}

Then the transfer rule \(t\) weakly implements \(\tau^{FB}\).
\end{proposition}

\begin{proof}
See Appendix.
\end{proof}

Proposition~\ref{prop:full_fb_implementability} separates implementation into two independent requirements. Condition~\eqref{eq:fb_support} guarantees that the mediator has no profitable deviation to any Bayes-plausible experiment, while condition~\(\int V_M(\mu;t)\, d\tau^{FB}(\mu) - c(\tau^{FB}) \ge \bar U\) guarantees acceptance of the contract.

Unlike Theorem~\ref{thm:fb_suff_support_alignment}, no global payoff alignment is required. Instead, implementation depends on whether an action-contingent transfer can modify the mediator's objective so that the desired experiment is supported by a subgradient of the information-cost functional.

The key restriction is that transfers depend only on observable receiver actions. Consequently, \(\mu \longmapsto t(a^*(\mu))\) is constant on each receiver-action region. The principal therefore cannot choose an arbitrary function on the belief simplex, but only one generated by the finite-dimensional transfer vector \(\bigl(t(a)\bigr)_{a\in A}\).

Whether such a transfer exists is precisely the implementability problem under delegated information acquisition.

\begin{remark}[Necessity under additional regularity]
\label{rem:fb_necessity}
Proposition~\ref{prop:full_fb_implementability} provides a sufficient condition for implementation. For continuous posterior-separable costs, the supporting inequality \eqref{eq:fb_support} is also necessary when the posterior cost index is chosen as the subgradient.
\end{remark}

For continuous posterior-separable costs, Theorem~\ref{thm:fb_suff_support_alignment} is a special case of Proposition~\ref{prop:full_fb_implementability}. When principal and mediator preferences satisfy \(V_M^0(\mu) = V_P^0(\mu)+\beta \qquad \forall\mu\in\Delta(\Omega)\), a flat transfer makes both players rank every Bayes-plausible experiment identically, so the required supporting condition follows directly.

\section{Second-Best as a Distorted Persuasion Problem}
\label{subsec:virtual_persuasion}

The previous section studied when the principal's first-best experiment can be decentralized despite delegated information acquisition. When first-best implementation fails, the principal faces a fundamentally different problem. Rather than choosing an information structure directly, she chooses a transfer rule that induces the mediator to select an experiment through its own optimization problem. The principal therefore controls information only indirectly through the mediator's incentives.

This section develops a reduced-form characterization of the second-best problem. The key observation is that delegation affects Bayesian persuasion in two conceptually distinct ways. First, not every Bayes-plausible experiment can necessarily be implemented through action-contingent transfers. Second, even among implementable experiments, different experiments require different transfer expenditures to induce. The first effect is a feasibility constraint; the second is an implementation-cost distortion.

The analysis therefore proceeds in two steps. We first define the set of experiments that can arise in equilibrium under some transfer rule. We then summarize the entire agency problem by the minimum expected transfer required to implement each experiment. This implementation-cost functional yields an exact reduced-form representation of the principal's optimization problem and makes precise how delegated information acquisition modifies the standard Bayesian persuasion benchmark.

\subsection{Implementable Experiments}
\label{subsec:implementable_experiments}

The first-best benchmark optimizes over the entire set of Bayes-plausible experiments. Under delegation, however, the principal cannot freely choose an experiment. Instead, the experiment must arise as the mediator's optimal response to some transfer rule. This motivates the following definition.

\begin{definition}[Implementable experiment]
\label{def:implementable_experiments}
An experiment \(\tau\in\mathcal K\) is \emph{implementable} if there exists a transfer rule \(t:A\rightarrow\mathbb R\) such that \(\tau \in \BR(t)\), and \(\int V_M(\mu;t) \,d\tau(\mu) - c(\tau) \ge \bar U\).

The set of implementable experiments is denoted by \(\mathcal I \subseteq \mathcal K\).
\end{definition}

The definition is intentionally behavioral. An experiment is implementable whenever it can arise as an equilibrium outcome of the delegated persuasion game under some transfer rule. No particular representation of the mediator's optimality conditions is imposed at this stage.

The set \(\mathcal I\) summarizes the first effect of delegation. Unlike the first-best problem, the principal's feasible set is no longer the entire Bayes-plausible set \(\mathcal K\), but only the subset of experiments that can be induced through observable action-contingent transfers.

The following proposition provides a convenient sufficient condition for implementation.

\begin{proposition}[Supporting sufficient condition]
\label{prop:implementation_support}
Let \(\tau\in\mathcal K\). Suppose there exist a transfer rule \(t:A\rightarrow\mathbb R\) and a subgradient \(\psi \in \partial c(\tau)\) such that \(\int V_M(\mu;t) \,d\tau(\mu) - c(\tau) \ge \bar U\), and

\begin{equation}
\label{eq:implementation_support} \int \left[ V_M(\mu;t) - \psi(\mu) \right] d\tau'(\mu) \le \int \left[ V_M(\mu;t) - \psi(\mu) \right] d\tau(\mu) \qquad \forall \tau' \in \mathcal K.
\end{equation}

Then \(\tau\) is implementable.
\end{proposition}

\begin{proof}
The proof of Proposition~\ref{prop:full_fb_implementability} applies with \(\tau\) in place of \(\tau^{FB}\).
\end{proof}

Proposition~\ref{prop:implementation_support} separates the implementation problem into two independent components.

Condition~\eqref{eq:implementation_support} guarantees that the mediator has no profitable deviation to any Bayes-plausible experiment, while condition~\(\int V_M(\mu;t) \,d\tau(\mu) - c(\tau) \ge \bar U\) guarantees that the mediator accepts the contract.

Unlike the first-best problem, however, the principal cannot choose an arbitrary perturbation of posterior payoffs. Transfers depend only on the receiver's realized action. Consequently, \(\mu \mapsto t(a^*(\mu))\) is constant on every receiver-action region. The principal therefore controls the mediator's incentives only through the finite-dimensional vector \(\bigl(t(a)\bigr)_{a\in A}\), rather than through an arbitrary function on the belief simplex. Whether this finite-dimensional transfer vector is rich enough to support a desired experiment is precisely the implementation problem under delegated persuasion.

\subsection{Implementation Costs}
\label{subsec:implementation_cost}

Even when an experiment is implementable, different experiments need not require the same transfer expenditure. Some experiments may be inexpensive to induce because the mediator's private incentives already point in the desired direction. Others may require substantial compensation to offset both information-acquisition costs and incentive misalignment.

The following definition summarizes this second source of distortion.

\begin{definition}[Minimum implementation cost]
\label{def:T_tau}

For every \(\tau\in\mathcal I\), define

\[
T(\tau) = \inf_t \left\{ \int t(a^*(\mu)) \,d\tau(\mu) : t \text{ implements } \tau \right\}.
\]

If \(\tau\notin\mathcal I\), set \(T(\tau)=+\infty\).
\end{definition}

The implementation-cost functional summarizes the entire agency problem from the principal's perspective.

The value \(T(\tau)\) is the minimum expected transfer expenditure required to induce the mediator to choose the experiment \(\tau\). Experiments that cannot be implemented receive infinite cost.

Consequently, delegation modifies the standard persuasion problem through two channels.

First, feasibility is restricted from \(\mathcal K\) to \(\mathcal I\).

Second, implementable experiments are no longer equally costly to induce. The function \(T(\tau)\) records exactly how expensive each implementable experiment is for the principal.

The next lemma shows that under unrestricted transfers the implementation-cost functional admits an exact closed-form representation. The result is a direct consequence of the fact that any slack in the mediator's participation constraint can be eliminated by uniformly shifting the transfer schedule without affecting incentive compatibility.

\begin{lemma}[Exact implementation cost]
\label{lem:participation_binds}

Suppose transfers are unrestricted.

For every \(\tau\in\mathcal I\), \(T(\tau) = \bar U - \int V_M^0(\mu) \,d\tau(\mu) + c(\tau)\).

Equivalently, every cost-minimizing implementing transfer leaves the mediator exactly indifferent between accepting and rejecting the contract.

\end{lemma}

\begin{proof}
See Appendix.
\end{proof}

Lemma~\ref{lem:participation_binds} is one of the central structural results of the paper. It shows that implementation costs consist of three economically distinct components.

The first compensates the mediator for its outside option. The second offsets the mediator's intrinsic valuation of the induced experiment. The third reimburses the direct cost of information acquisition.

Thus, once implementation is feasible, the entire agency problem is summarized by a single scalar function \(T(\tau)\), which serves as the endogenous cost of delegated information acquisition.

\subsection{Properties of the Implementation-Cost Functional}
\label{subsec:implementation_cost_properties}

The implementation-cost functional introduced above summarizes the entire agency friction associated with delegated information acquisition. This subsection establishes its basic properties. The first result shows that implementation costs inherit the geometry of the underlying information-cost technology while additionally reflecting the mediator's intrinsic preferences over information.

\begin{proposition}[Properties of the implementation-cost functional]
\label{prop:T_properties}

Suppose transfers are unrestricted.

\begin{enumerate}[label=(\roman*)]

\item
For every \(\tau\in\mathcal I\), \(T(\tau) = \bar U - \int V_M^0(\mu) \,d\tau(\mu) + c(\tau)\).
\item
If \(V_M^0\) is Bayes-linear and \(c\) is monotone with respect to the Blackwell order, then \(\tau' \succeq_B \tau \quad \Longrightarrow \quad T(\tau') \ge T(\tau)\), for every \(\tau,\tau'\in\mathcal I\).

\item
Under the conditions of (ii), more informative experiments require weakly greater transfer expenditure.

\item
If \(V_M^0\) is convex, posterior dispersion raises the mediator's intrinsic valuation of information. If \(V_M^0\) is concave, posterior dispersion lowers it. If \(V_M^0\) is Bayes-linear, the intrinsic payoff contributes the same constant to every Bayes-plausible experiment.

\end{enumerate}

\end{proposition}

\begin{proof}
Part (i) follows from Lemma~\ref{lem:participation_binds}.

For part (ii), Bayes-linearity implies \(\int V_M^0(\mu) \,d\tau(\mu) = V_M^0(\mu_0) = \int V_M^0(\mu) \,d\tau'(\mu)\) for every \(\tau,\tau'\in\mathcal K\).

Consequently, \(T(\tau') - T(\tau) = c(\tau') - c(\tau)\).

Blackwell monotonicity of \(c\) therefore implies \(T(\tau') \ge T(\tau)\).

Part (iii) follows immediately.

Part (iv) is Jensen's inequality.
\end{proof}

Proposition~\ref{prop:T_properties} shows that implementation costs are determined by two conceptually distinct forces.

The first is technological. If \(c\) is Blackwell monotone, more informative experiments require weakly greater information-acquisition costs.

The second is incentive-based. Whenever the mediator derives intrinsic utility from information, implementation costs are shifted upward or downward depending on how the mediator values posterior dispersion.

The Bayes-linear case provides an important benchmark. Since the mediator's intrinsic payoff is then constant across all Bayes-plausible experiments, implementation costs depend only on the outside option and the information technology.

\subsection{General Second-Best Characterization}
\label{subsec:second_best_general}

We now derive the principal's optimization problem in reduced form.

The first-best benchmark optimizes directly over the entire set of Bayes-plausible experiments. Delegation modifies this problem in two ways.

First, the feasible set becomes \(\mathcal I \subseteq \mathcal K\), because only implementable experiments can arise in equilibrium.

Second, each implementable experiment carries an endogenous transfer cost summarized by \(T(\tau)\).

The following theorem shows that these two objects completely characterize the principal's second-best problem.

\begin{theorem}[General second-best characterization]
\label{thm:secondbest_general}

Suppose the principal's optimization problem admits an optimum.

Then

\begin{equation}
\label{eq:SB_value} W^{SB} = \sup_{\tau\in\mathcal I} \left\{ \int V_P^0(\mu) \,d\tau(\mu) - T(\tau) \right\}.
\end{equation}

Moreover,

\begin{equation}
\label{eq:SB_argmax} \tau^{SB} \in \argmax_{\tau\in\mathcal I} \left\{ \int V_P^0(\mu) \,d\tau(\mu) - T(\tau) \right\}.
\end{equation}

\end{theorem}

\begin{proof}
See Appendix.
\end{proof}

Theorem~\ref{thm:secondbest_general} is the central representation result of the paper.

Unlike the original delegated persuasion problem, which is a bilevel optimization problem over transfers and experiments, the reduced problem is a single optimization over experiments.

Delegation therefore affects Bayesian persuasion only through two objects: \((\mathcal I,T)\).

The set \(\mathcal I\) captures implementability.

The function \(T\) captures the endogenous cost of inducing each implementable experiment.

No other aspect of the delegated contracting problem enters the principal's objective.

\subsection{Joint-Surplus Representation}
\label{subsec:joint_surplus}

The reduced-form representation becomes especially transparent under unrestricted transfers.

Substituting the implementation-cost formula of Lemma~\ref{lem:participation_binds} into Theorem~\ref{thm:secondbest_general} immediately yields the following characterization.

\begin{corollary}[Joint-surplus representation]
\label{cor:joint_surplus}

Suppose transfers are unrestricted and a second-best optimum exists.

Then

\[
\tau^{SB} \in \argmax_{\tau\in\mathcal I} \left\{ \int \bigl[ V_P^0(\mu) + V_M^0(\mu) \bigr] \,d\tau(\mu) - c(\tau) \right\}.
\]

\end{corollary}

\begin{proof}

By Lemma~\ref{lem:participation_binds}, \(T(\tau) = \bar U - \int V_M^0(\mu) \,d\tau(\mu) + c(\tau)\).

Substituting this expression into \eqref{eq:SB_argmax} gives \(\int V_P^0 \,d\tau - \bar U + \int V_M^0 \,d\tau - c(\tau)\).

Since \(\bar U\) is constant, it does not affect the maximization.

The result follows.

\end{proof}

Corollary~\ref{cor:joint_surplus} provides a particularly simple economic interpretation.

Once transfers are chosen optimally, the principal internalizes the mediator's intrinsic payoff exactly. The remaining distortion arises not from conflicting objectives but from the restriction that optimization is carried out over the implementable set \(\mathcal I\) rather than over the entire Bayes-plausible set \(\mathcal K\).

Relative to the first-best objective, delegation adds the mediator's intrinsic payoff and restricts optimization to \(\mathcal I\).

Taken together, these two objects, \(\mathcal I\) and \(T(\tau)\), completely summarize the economic consequences of delegated information acquisition.

\subsection{Delegation Irrelevance}
\label{subsec:delegation_irrelevance}

The implementation-cost representation provides a simple benchmark for understanding when delegation has no effect on information design. Delegation is irrelevant whenever the principal's preferred experiment remains implementable and the implementation cost does not alter the ranking of Bayes-plausible experiments.

The following proposition formalizes this observation.

\begin{proposition}[Delegation irrelevance]
\label{prop:delegation_irrelevance}

Suppose \(\tau^{FB}\in\mathcal I\), and suppose further that \(V_M^0\) is Bayes-linear.

Then

\[
\tau^{FB} \in \argmax_{\tau\in\mathcal I} \left\{ \int \bigl[ V_P^0(\mu) + V_M^0(\mu) \bigr] \,d\tau(\mu) - c(\tau) \right\}.
\]

Consequently, \(\tau^{FB}\) is second-best optimal.

\end{proposition}

\begin{proof}

Since \(V_M^0\) is Bayes-linear, \(\int V_M^0(\mu) \,d\tau(\mu) = V_M^0(\mu_0)\) for every \(\tau\in\mathcal K\).

Hence the mediator's intrinsic payoff contributes the same constant to every Bayes-plausible experiment.

The joint-surplus objective therefore differs from the first-best objective only by an additive constant, \(V_M^0(\mu_0)\), which does not affect the ranking of experiments.

Since \(\tau^{FB}\in\mathcal I\), it is feasible for the second-best problem and therefore remains optimal.

\end{proof}

Proposition~\ref{prop:delegation_irrelevance} identifies the economically most transparent benchmark. Whenever the mediator's intrinsic payoff is Bayes-linear and \(\tau^{FB}\) remains implementable, delegation does not distort information acquisition. The only remaining agency friction is implementability itself.

Notice that the proposition deliberately concludes that \(\tau^{FB}\) is second-best optimal rather than asserting \(\tau^{SB}=\tau^{FB}\). Without uniqueness of the optimizer, several second-best experiments may exist.

\subsection{Comparative Statics}
\label{subsec:general_distortion}

The implementation-cost representation also provides immediate comparative statics on how mediator preferences distort information acquisition.

The key observation is that the mediator's intrinsic payoff enters the second-best objective exactly through the joint-surplus term \(V_P^0+V_M^0\).

Its curvature therefore changes the relative value of posterior dispersion without changing the underlying information technology.

\begin{proposition}[Curvature and posterior dispersion]
\label{prop:general_distortion}

Suppose transfers are unrestricted.

\begin{enumerate}[label=(\roman*)]

\item
If \(V_M^0\) is convex,

then for every pair of Bayes-plausible experiments \(\tau,\tau'\)

with \(\tau'\) a mean-preserving spread of \(\tau\), \(\int V_M^0 \,d\tau' \ge \int V_M^0 \,d\tau\).
\item

If \(V_M^0\) is concave,

the inequality is reversed.

\item

If \(V_M^0\) is Bayes-linear, \(\int V_M^0 \,d\tau = V_M^0(\mu_0)\) for every Bayes-plausible experiment.

\end{enumerate}

\end{proposition}

\begin{proof}

Each statement follows directly from Jensen's inequality.

\end{proof}

Proposition~\ref{prop:general_distortion} identifies the direction in which mediator preferences modify the principal's objective.

Convex preferences increase the value of posterior dispersion.

Concave preferences decrease it.

Bayes-linear preferences are neutral.

Importantly, the proposition concerns the objective function itself. It does not imply that the resulting optimal experiments are ordered by the Blackwell relation.

Indeed, two different objective functions may generate optimal experiments that are incomparable under the Blackwell order, especially when the state space contains more than two states. Additional structure is therefore required before one can compare the information content of first-best and second-best experiments.

\subsection{Unified Interpretation}
\label{subsec:unified}

The results of this section admit a simple interpretation.

The first-best benchmark solves

\[
\max_{\tau\in\mathcal K} \left\{ \int V_P^0 \,d\tau - c(\tau) \right\},
\]

whereas delegated persuasion solves

\[
\max_{\tau\in\mathcal I} \left\{ \int V_P^0 \,d\tau - T(\tau) \right\}.
\]

Thus delegation modifies the standard persuasion problem through two objects only: \((\mathcal I,T)\).

The feasible set changes from \(\mathcal K\) to \(\mathcal I\), reflecting implementability.

The objective changes from \(c(\tau)\) to \(T(\tau)\), reflecting the endogenous cost of inducing information acquisition.

Under unrestricted transfers, \(T(\tau) = \bar U - \int V_M^0 \,d\tau + c(\tau)\), so the implementation-cost distortion decomposes into three economically distinct components:

\begin{enumerate}

\item
the mediator's outside option;

\item
the mediator's intrinsic valuation of information;

\item
the technological cost of information acquisition.

\end{enumerate}

This decomposition isolates the exact channels through which delegation alters Bayesian persuasion. The following subsection specializes this representation to entropy costs.

\subsection{Entropy Costs}
\label{subsec:entropy_geometry}

Suppose transfers are unrestricted and \(\Omega=\{0,1\}\), and identify a posterior with \(\mu\in[0,1]\). Let \(c(\tau)=H(\mu_0)-\int H\,d\tau\), where \(\mu_0\in(0,1)\) and \(H\) is continuous and strictly concave on \([0,1]\) and continuously differentiable on \((0,1)\). Write \(W=V_P^0+V_M^0\). Suppose \(W\) is continuous and affine on each side of \(\bar\mu\in(0,1)\), with slopes \(m_-\) and \(m_+\), and define \(\Delta_W=m_+-m_-\). Maximizing \(\int(W+H)\,d\tau\) over \(\mathcal K\) relaxes the joint-surplus problem over \(\mathcal I\).

\begin{proposition}[Entropy tangency]
\label{prop:entropy_tangency}
Suppose \(\tau^{SB}=p^{SB}\delta_{\mu_H^{SB}}+(1-p^{SB})\delta_{\mu_L^{SB}}\) maximizes \(\int(W+H)\,d\tau\) over \(\mathcal K\) and belongs to \(\mathcal I\), where \(0<\mu_L^{SB}<\bar\mu<\mu_H^{SB}<1\) and \(\mu_L^{SB}<\mu_0<\mu_H^{SB}\). Then it is second-best optimal, \(p^{SB}=(\mu_0-\mu_L^{SB})/(\mu_H^{SB}-\mu_L^{SB})\), and a common slope \(\ell\) satisfies \(\ell=\frac{W(\mu_H^{SB})+H(\mu_H^{SB})-W(\mu_L^{SB})-H(\mu_L^{SB})}{\mu_H^{SB}-\mu_L^{SB}}\), \(m_-+H'(\mu_L^{SB})=\ell\), \(m_++H'(\mu_H^{SB})=\ell\).

Consequently, \(H'(\mu_L^{SB})-H'(\mu_H^{SB})=\Delta_W\).
\end{proposition}

\begin{corollary}[Shannon odds-ratio relation]
\label{cor:shannon_explicit}
Under Shannon entropy \(H(\mu)=-\mu\ln\mu-(1-\mu)\ln(1-\mu)\), with \(0\ln0=0\), the endpoints in Proposition~\ref{prop:entropy_tangency} satisfy \(\frac{(1-\mu_L^{SB})\mu_H^{SB}}{\mu_L^{SB}(1-\mu_H^{SB})}=e^{\Delta_W}\), \(\mu_H^{SB}=\frac{\mu_L^{SB}}{\mu_L^{SB}+(1-\mu_L^{SB})e^{-\Delta_W}}\).

The remaining tangency and secant conditions and Bayes plausibility must also hold.
\end{corollary}

\begin{corollary}[Implementation costs and mutual information]
\label{cor:persuasion_rent_entropy}
Under Shannon entropy, for every \(\tau\in\mathcal I\) and any signal \(S\) inducing \(\tau\), \(c(\tau)=I(\omega;S)\) and \(T(\tau)=\bar U-\int V_M^0\,d\tau+I(\omega;S)\).
\end{corollary}

\section{Conclusion}
\label{sec:conclusion}

This paper studies Bayesian persuasion when information acquisition is delegated to an intermediary whose experiment choice is privately selected and costly. Because the principal cannot observe or contract directly on the mediator's information acquisition decision, persuasion becomes an incentive problem in which information structures must be induced indirectly through action-contingent transfers.

The analysis develops a unified framework for studying delegated information acquisition. We first derive geometric sufficient conditions under which the principal's preferred experiment can be implemented despite hidden information acquisition. We then characterize the delegated second-best through an implementation-cost representation that summarizes the agency friction by the minimum expected transfer required to induce each implementable experiment. Under unrestricted transfers, this implementation cost admits an exact closed-form expression, reducing the delegated persuasion problem to joint-surplus maximization over the set of implementable experiments.

The implementation-cost representation reveals that delegation modifies Bayesian persuasion through two distinct channels. First, not every Bayes-plausible experiment is implementable through observable action-contingent transfers. Second, even among implementable experiments, different information structures require different transfer expenditures. Together, these two objects—the implementable set and the implementation-cost functional—completely summarize the economic consequences of delegated information acquisition.

The framework also provides comparative statics on the role of the mediator's intrinsic preferences. Convex preferences increase the value of posterior dispersion, concave preferences decrease it, and Bayes-linear preferences leave the ranking of Bayes-plausible experiments unchanged whenever the principal's preferred experiment remains implementable. Under entropy costs, we derive tangency conditions for implementable interior two-posterior maximizers of the unrestricted joint-surplus problem and a Shannon odds-ratio relation, together with an interpretation of information costs through mutual information.

More broadly, the paper establishes a connection between Bayesian persuasion, delegated information acquisition, and incentive theory. By introducing implementation costs as the central object linking persuasion and contracting, the framework provides a tractable foundation for studying delegated information production in environments such as financial intermediation, AI auditing, consulting, certification, and regulatory oversight.

\section*{Appendix: Proofs}
\label{sec:appendix}
\subsection*{Proofs for Section~\ref{sec:fb_benchmarks}}

\begin{proof}[Proof of Theorem~\ref{thm:fb_suff_support_alignment}]
Let \(\tau^{FB}\) solve the first-best problem
\[
\max_{\tau\in\mathcal K} \left\{ \int V_P^0(\mu)\,d\tau(\mu)-c(\tau) \right\}.
\]
Suppose that \(V_M^0(\mu)=V_P^0(\mu)+\beta \qquad \forall \mu\in\Delta(\Omega)\).

Let \(\tilde t(a):=\bar t\) for every \(a\in A\), for a constant \(\bar t\) to be fixed below. For \emph{every} \(\tau\in\mathcal K\) (not merely those whose support induces the same actions as \(\tau^{FB}\)), since \(\tilde t(a^*(\mu))=\bar t\) identically in \(\mu\),
\[
\int V_M(\mu;\tilde t)\,d\tau(\mu)-c(\tau) = \int \left[ V_M^0(\mu)+\bar t \right] \,d\tau(\mu)-c(\tau) = \int V_P^0(\mu)\,d\tau(\mu)-c(\tau)+\beta+\bar t,
\]
using \(V_M^0=V_P^0+\beta\), which holds at every \(\mu\in\Delta(\Omega)\) and therefore applies regardless of which posteriors \(\tau\) places mass on. The final two terms are constants independent of \(\tau\). Hence the mediator ranks \emph{every} experiment in \(\mathcal K\) exactly as the principal ranks it in the first-best problem, and in particular
\[
\int V_M(\mu;\tilde t)\,d\tau(\mu)-c(\tau) \le \int V_M(\mu;\tilde t)\,d\tau^{FB}(\mu)-c(\tau^{FB}) \qquad \forall\tau\in\mathcal K,
\]
so \(\tau^{FB}\) is optimal for the mediator under \(\tilde t\).

It remains to satisfy participation. Let \(U_M(\tilde t,\tau^{FB}) := \int V_M(\mu;\tilde t)\,d\tau^{FB}(\mu)-c(\tau^{FB})\).

For any scalar \(r\in\mathbb R\), define the shifted transfer \(\tilde t^r(a):=\tilde t(a)+r=\bar t+r \qquad \forall a\in A\), which remains flat. This shift adds \(r\) to the mediator's payoff under every experiment and subtracts \(r\) from the principal's payoff, but it does not change the mediator's ranking of experiments, since it too is constant across \(\mu\). Choose \(r\ge \bar U-U_M(\tilde t,\tau^{FB})\).

Then \(\int V_M(\mu;\tilde t^r)\,d\tau^{FB}(\mu)-c(\tau^{FB}) \ge \bar U\).

Thus \(\tau^{FB}\) is optimal for the mediator under the flat transfer \(\tilde t^r\) and satisfies participation. Hence \(\tau^{FB}\) is implementable.
\end{proof}

\begin{proof}[Proof of Remark~\ref{rem:alpha_one}]
Suppose instead that \(V_M^0(\mu)=\alpha V_P^0(\mu)+\beta\) for some \(\alpha>0\). Ignoring constants and any transfer component that is independent of the experiment, the mediator's objective becomes \(\int \alpha V_P^0(\mu)\,d\tau(\mu)-c(\tau)\).

Multiplying by \(1/\alpha>0\), this is equivalent to \(\int V_P^0(\mu)\,d\tau(\mu)-\frac{1}{\alpha}c(\tau)\).

The first-best objective is instead \(\int V_P^0(\mu)\,d\tau(\mu)-c(\tau)\).

Unless \(\alpha=1\), the relative weight on the information cost differs. Therefore the two maximization problems generally have different solutions. Hence proportionality is not enough; exact affine alignment with unit coefficient is the condition that preserves the principal's cost-adjusted ranking of experiments.
\end{proof}

\begin{proof}[Proof of Proposition~\ref{prop:full_fb_implementability}]
Suppose there exist \(t\) and \(\psi^{FB}\in\partial c(\tau^{FB})\) such that participation holds and for every \(\tau\in\mathcal K\),
\[
\int \left[ V_M(\mu;t)-\psi^{FB}(\mu) \right] d\tau(\mu) \le \int \left[ V_M(\mu;t)-\psi^{FB}(\mu) \right] d\tau^{FB}(\mu).
\]
Because \(\psi^{FB}\in\partial c(\tau^{FB})\), \(c(\tau) \ge c(\tau^{FB}) + \int \psi^{FB}(\mu)\, d(\tau-\tau^{FB})(\mu) \qquad \forall \tau\in\mathcal K\).

Rearranging, \(-c(\tau) \le -c(\tau^{FB}) - \int \psi^{FB}(\mu)\,d\tau(\mu) + \int \psi^{FB}(\mu)\,d\tau^{FB}(\mu)\).

Combining this inequality with the assumed payoff-index inequality implies \(\int V_M(\mu;t)\,d\tau(\mu)-c(\tau) \le \int V_M(\mu;t)\,d\tau^{FB}(\mu)-c(\tau^{FB})\) for every \(\tau\in\mathcal K\). Therefore
\[
\tau^{FB} \in \argmax_{\tau\in\mathcal K} \left\{ \int V_M(\mu;t)\,d\tau(\mu)-c(\tau) \right\}.
\]
Since participation also holds by assumption, \(\tau^{FB}\) is implemented by \(t\).

The equivalence with the standard persuasion formulation follows immediately from the inequality
\[
\int \left[ V_M(\mu;t)-\psi^{FB}(\mu) \right] d\tau(\mu) \le \int \left[ V_M(\mu;t)-\psi^{FB}(\mu) \right] d\tau^{FB}(\mu),
\]
which says precisely that \(\tau^{FB}\) maximizes the expected payoff index \(\mu\mapsto V_M(\mu;t)-\psi^{FB}(\mu)\) over \(\mathcal K\).
\end{proof}

\begin{proof}[Proof of Lemma~\ref{lem:participation_binds}]
Suppose \(t\) implements \(\tau\) with \(\int V_M(\mu;t)\,d\tau(\mu)-c(\tau) > \bar U\), and let \(\varepsilon:=\left[\int V_M(\mu;t)\,d\tau(\mu)-c(\tau)\right]-\bar U>0\).

Define \(t'(a):=t(a)-\varepsilon\) for every \(a\in A\). For any \(\tau'\in\mathcal K\), \(\int V_M(\mu;t')\,d\tau'(\mu)-c(\tau') = \int V_M(\mu;t)\,d\tau'(\mu)-c(\tau')-\varepsilon\), since \(t'(a^*(\mu))=t(a^*(\mu))-\varepsilon\) identically in \(\mu\). Subtracting the same constant \(\varepsilon\) from the mediator's payoff at every \(\tau'\in\mathcal K\) leaves the ranking of experiments, and hence the argmax correspondence, unchanged; in particular \(\tau\) remains optimal for the mediator under \(t'\), so \(t'\) also implements \(\tau\). Participation under \(t'\) holds with equality: \(\int V_M(\mu;t')\,d\tau(\mu)-c(\tau) = \left[\int V_M(\mu;t)\,d\tau(\mu)-c(\tau)\right]-\varepsilon = \bar U\).

Finally,
\[
\int t'(a^*(\mu))\,d\tau(\mu) = \int t(a^*(\mu))\,d\tau(\mu)-\varepsilon < \int t(a^*(\mu))\,d\tau(\mu),
\]
so \(t\) does not attain the infimum in Definition~\ref{def:T_tau}. Participation implies that every implementing transfer has expected payment at least \(\bar U-\int V_M^0\,d\tau+c(\tau)\). Starting from any implementing transfer, either participation already binds or the construction above yields an implementing transfer attaining this lower bound. Hence the infimum is attained. At any minimizing transfer, participation holds with equality, \(\int V_M(\mu;t)\,d\tau(\mu)-c(\tau)=\bar U\), i.e. \(\int t(a^*(\mu))\,d\tau(\mu) = \bar U-\int V_M^0(\mu)\,d\tau(\mu)+c(\tau)\), which is exactly \(T(\tau)=\bar U-\int V_M^0\,d\tau+c(\tau)\).
\end{proof}

\begin{proof}[Proof of Proposition~\ref{prop:T_properties}]
Part (i) is Lemma~\ref{lem:participation_binds}.

For part~(ii), suppose \(c\) is Blackwell monotone and \(V_M^0\) is Bayes-linear, and let \(\tau,\tau'\in\mathcal I\) be experiments at which the respective cost-minimizing transfers are attained. By part~(i), \(T(\tau)=\bar U-\int V_M^0(\mu)\,d\tau(\mu)+c(\tau), \qquad T(\tau')=\bar U-\int V_M^0(\mu)\,d\tau'(\mu)+c(\tau')\).

Since \(V_M^0\) is Bayes-linear, \(\int V_M^0(\mu)\,d\tau(\mu) = V_M^0(\mu_0) = \int V_M^0(\mu)\,d\tau'(\mu)\).

Therefore, \(T(\tau')-T(\tau)=c(\tau')-c(\tau)\).

If \(\tau'\succeq_B\tau\), Blackwell monotonicity of \(c\) implies \(c(\tau')\ge c(\tau)\), and hence \(T(\tau')\ge T(\tau)\).

Part~(iii) follows immediately from part~(ii).

For part~(iv), fix any \(\tau\in\mathcal I\) at which the infimum is attained. Part~(i) gives \(T(\tau) = \bar U-\int V_M^0(\mu)\,d\tau(\mu)+c(\tau)\).

Substituting this expression into the reduced second-best objective \(\int V_P^0(\mu)\,d\tau(\mu)-T(\tau)\) yields \(\int V_P^0(\mu)\,d\tau(\mu) - \left[ \bar U-\int V_M^0(\mu)\,d\tau(\mu)+c(\tau) \right]\), or equivalently \(\int \bigl[ V_P^0(\mu)+V_M^0(\mu) \bigr] \,d\tau(\mu) - c(\tau) - \bar U\).

This proves the objective representation.

The curvature statement follows from Jensen's inequality. If \(V_M^0\) is convex, then for every Bayes-plausible experiment \(\tau\in\mathcal K\), \(\int V_M^0(\mu)\,d\tau(\mu) \ge V_M^0\left(\int \mu\,d\tau(\mu)\right) = V_M^0(\mu_0)\).

Thus the mediator's intrinsic payoff weakly rewards posterior dispersion. If \(V_M^0\) is concave, Jensen's inequality is reversed, so posterior dispersion is weakly penalized.

Finally, when \(V_M^0\) is Bayes-linear, part~(ii) applies to entropy costs, where \(c(\tau)=H(\mu_0)-\int H(\mu)\,d\tau(\mu)\).

If \(\tau'\succeq_B\tau\), then \(\tau'\) is more informative, so it has weakly lower expected posterior entropy: \(\int H(\mu)\,d\tau'(\mu) \le \int H(\mu)\,d\tau(\mu)\).

Hence \(c(\tau')\ge c(\tau)\).

Since \(V_M^0\) is Bayes-linear, the intrinsic payoff term is constant across Bayes-plausible experiments, exactly as in part (ii). Therefore, for \(\tau,\tau'\in\mathcal I\) at which the respective infima are attained, \(T(\tau')-T(\tau)=c(\tau')-c(\tau)\ge0\).

Thus more informative experiments require weakly larger implementation transfers.
\end{proof}

\begin{proof}[Proof of Theorem~\ref{thm:secondbest_general}]
Let
\[
P(t,\tau) := \int V_P(\mu;t)\,d\tau(\mu) = \int V_P^0(\mu)\,d\tau(\mu) - \int t(a^*(\mu))\,d\tau(\mu)
\]
denote the principal's payoff when the transfer rule is \(t\) and the mediator chooses \(\tau\).

By definition, the principal's second-best value is
\[
W^{SB} = \sup_{t,\tau} \left\{ P(t,\tau): t\text{ implements }\tau \right\}.
\]
Equivalently,
\[
W^{SB} = \sup_{t,\tau} \left\{ \int V_P^0(\mu)\,d\tau(\mu) - \int t(a^*(\mu))\,d\tau(\mu) : t \text{ implements } \tau \right\}.
\]

We now rewrite this problem experiment by experiment. If a transfer rule \(t\) implements \(\tau\), then by definition \(\tau\in\mathcal I\). Conversely, if \(\tau\in\mathcal I\), then there exists at least one transfer rule implementing \(\tau\). Thus the principal's problem can be written as
\[
W^{SB} = \sup_{\tau\in\mathcal I} \sup_{t:\,t\text{ implements }\tau} \left\{ \int V_P^0(\mu)\,d\tau(\mu) - \int t(a^*(\mu))\,d\tau(\mu) \right\}.
\]
For a fixed implementable experiment \(\tau\), the term \(\int V_P^0(\mu)\,d\tau(\mu)\) does not depend on the transfer rule. Therefore,
\[
\begin{aligned} \sup_{t:\,t\text{ implements }\tau} \left\{ \int V_P^0\,d\tau - \int t(a^*(\mu))\,d\tau(\mu) \right\} &= \int V_P^0\,d\tau \\ &\quad - \inf_{t:\,t\text{ implements }\tau} \int t(a^*(\mu))\,d\tau(\mu). \end{aligned}
\]
By the definition of the implementation-cost functional,
\[
T(\tau) = \inf_{t:\,t\text{ implements }\tau} \int t(a^*(\mu))\,d\tau(\mu).
\]
Hence
\[
W^{SB} = \sup_{\tau\in\mathcal I} \left\{ \int V_P^0(\mu)\,d\tau(\mu)-T(\tau) \right\}.
\]
This proves the value representation.

It remains to show that any second-best experiment maximizes the reduced objective. Let \((t^{SB},\tau^{SB})\) be a second-best optimum. Since \(t^{SB}\) implements \(\tau^{SB}\), we have \(\tau^{SB}\in\mathcal I\). Moreover, \(T(\tau^{SB}) \le \int t^{SB}(a^*(\mu))\,d\tau^{SB}(\mu)\), because \(T(\tau^{SB})\) is the infimum over all transfers that implement \(\tau^{SB}\).

Suppose, toward a contradiction, that \(\tau^{SB}\) does not maximize the reduced problem over \(\mathcal I\). Then there exists \(\hat\tau\in\mathcal I\) such that \(\int V_P^0(\mu)\,d\hat\tau(\mu)-T(\hat\tau) > \int V_P^0(\mu)\,d\tau^{SB}(\mu)-T(\tau^{SB})\).

By definition of \(T(\hat\tau)\), for every \(\varepsilon>0\) there exists a transfer rule \(\hat t_\varepsilon\) implementing \(\hat\tau\) such that \(\int \hat t_\varepsilon(a^*(\mu))\,d\hat\tau(\mu) \le T(\hat\tau)+\varepsilon\).

Choose \(\varepsilon>0\) small enough that
\[
\int V_P^0\,d\hat\tau - T(\hat\tau) - \varepsilon > \int V_P^0\,d\tau^{SB} - \int t^{SB}(a^*(\mu))\,d\tau^{SB}(\mu).
\]
Then the principal's payoff from the feasible pair \((\hat t_\varepsilon,\hat\tau)\) satisfies
\[
\begin{aligned} P(\hat t_\varepsilon,\hat\tau) &= \int V_P^0\,d\hat\tau - \int \hat t_\varepsilon(a^*(\mu))\,d\hat\tau(\mu) \\ &\ge \int V_P^0\,d\hat\tau - T(\hat\tau) - \varepsilon \\ &> \int V_P^0\,d\tau^{SB} - \int t^{SB}(a^*(\mu))\,d\tau^{SB}(\mu) \\ &= P(t^{SB},\tau^{SB}). \end{aligned}
\]
This contradicts the optimality of \((t^{SB},\tau^{SB})\). Therefore,
\[
\tau^{SB} \in \argmax_{\tau\in\mathcal I} \left\{ \int V_P^0(\mu)\,d\tau(\mu)-T(\tau) \right\}.
\]

The theorem follows.
\end{proof}

\begin{proof}[Proof of Corollary~\ref{cor:joint_surplus}]
By Theorem~\ref{thm:secondbest_general}, every second-best experiment solves
\[
\tau^{SB} \in \argmax_{\tau\in\mathcal I} \left\{ \int V_P^0(\mu)\,d\tau(\mu)-T(\tau) \right\}.
\]

By Lemma~\ref{lem:participation_binds}, \(T(\tau) = \bar U - \int V_M^0(\mu)\,d\tau(\mu) + c(\tau) \qquad \forall \tau \in \mathcal I\), not merely at the optimizer. Substituting this expression into the reduced second-best objective, for every \(\tau\in\mathcal I\), gives
\[
\begin{aligned} \int V_P^0(\mu)\,d\tau(\mu)-T(\tau) &= \int V_P^0(\mu)\,d\tau(\mu) - \left[ \bar U - \int V_M^0(\mu)\,d\tau(\mu) + c(\tau) \right] \\ &= \int \bigl[ V_P^0(\mu)+V_M^0(\mu) \bigr] \,d\tau(\mu) - c(\tau) - \bar U. \end{aligned}
\]
Since this substitution is valid uniformly across \(\mathcal I\) and \(\bar U\) is constant in \(\tau\), maximizing the left-hand side over \(\mathcal I\) is equivalent to maximizing the joint-surplus expression on the right. Therefore
\[
\tau^{SB} \in \argmax_{\tau\in\mathcal I} \left\{ \int \bigl[ V_P^0(\mu)+V_M^0(\mu) \bigr] \,d\tau(\mu) - c(\tau) \right\}.
\]
\end{proof}

\begin{proof}[Proof of Proposition~\ref{prop:general_distortion}]
Since every \(\tau\in\mathcal K\) satisfies \(\int\mu\,d\tau(\mu)=\mu_0\), and \(\tau'\) a mean-preserving spread of \(\tau\) means \(\tau'\) is a garbling-free increase in risk holding the mean fixed, Jensen's inequality applied to the convex function \(V_M^0\) gives \(\int V_M^0\,d\tau'\ge\int V_M^0\,d\tau\) whenever \(\tau'\) is a mean-preserving spread of \(\tau\), since a mean-preserving spread is, by definition, obtained from \(\tau\) by moving mass apart while preserving the mean, and a convex function's integral against such a spread can only rise (this is the standard characterization of convex functions as those that are monotone under the convex, equivalently mean-preserving-spread, order). If \(V_M^0\) is concave the inequality reverses by the same argument applied to \(-V_M^0\). If \(V_M^0\) is affine, \(V_M^0(\mu)=a+b\cdot\mu\) for some \(a,b\), so \(\int V_M^0(\mu)\,d\tau(\mu)=a+b\cdot\int\mu\,d\tau(\mu)=a+b\cdot\mu_0= V_M^0(\mu_0)\) for every \(\tau\in\mathcal K\), using Bayes plausibility; this constant does not affect the ranking of experiments over any common feasible set.
\end{proof}

\begin{proof}[Proof of Proposition~\ref{prop:entropy_tangency}]
By Lemma~\ref{lem:participation_binds}, when participation binds the second-best problem over \(\mathcal I\) is equivalent to maximizing \(\int \bigl[ W(\mu)+H(\mu) \bigr] \,d\tau(\mu)\) over \(\tau\in\mathcal I\), up to constants, where \(W(\mu)=V_P^0(\mu)+V_M^0(\mu)\).

Let \(G(\mu):=W(\mu)+H(\mu)\).

Because \(W\) is affine on each side of \(\bar\mu\) and \(H\) is strictly concave, \(G\) is strictly concave on each side of \(\bar\mu\).

By hypothesis, \(\tau^{SB}\) maximizes \(\int G\,d\tau\) over \(\mathcal K\) and belongs to \(\mathcal I\). It therefore also maximizes this objective over \(\mathcal I\), and Lemma~\ref{lem:participation_binds} supplies an implementing transfer attaining the maximal principal payoff. A supporting affine majorant of the concave envelope of \(G\) at \(\mu_0\) agrees with \(G\) at both posterior support points: its expected gap from \(G\) is zero at the relaxed optimum, and both posterior weights are positive. Thus the line through the two support points majorizes \(G\). Write \(\tau^{SB}\) for this maximizer, \(\tau^{SB} = p^{SB}\delta_{\mu_H^{SB}} + (1-p^{SB})\delta_{\mu_L^{SB}}\).

Let \(\ell\) denote the slope of the supporting chord joining \((\mu_L^{SB},G(\mu_L^{SB})) \) and \( (\mu_H^{SB},G(\mu_H^{SB}))\).

Then \(\ell = \frac{ G(\mu_H^{SB})-G(\mu_L^{SB}) }{ \mu_H^{SB}-\mu_L^{SB} }\).

Substituting \(G=W+H\) gives
\[
\ell = \frac{ \bigl(W(\mu_H^{SB})+H(\mu_H^{SB})\bigr) - \bigl(W(\mu_L^{SB})+H(\mu_L^{SB})\bigr) }{ \mu_H^{SB}-\mu_L^{SB} },
\]
which proves~\(\ell=\frac{W(\mu_H^{SB})+H(\mu_H^{SB})-W(\mu_L^{SB})-H(\mu_L^{SB})}{\mu_H^{SB}-\mu_L^{SB}}\).

Because the solution is interior, the supporting chord is tangent to \(G\) at both posterior support points. On the left of \(\bar\mu\), \(W\) has slope \(m_-\), so \(G'(\mu_L^{SB}) = m_-+H'(\mu_L^{SB})\).

Tangency at the left posterior therefore gives \(m_-+H'(\mu_L^{SB})=\ell\).

Similarly, on the right of \(\bar\mu\), \(W\) has slope \(m_+\), so \(G'(\mu_H^{SB}) = m_++H'(\mu_H^{SB})\), and tangency at the right posterior gives \(m_++H'(\mu_H^{SB})=\ell\).

This proves~\(m_-+H'(\mu_L^{SB})=\ell\) and \(m_++H'(\mu_H^{SB})=\ell\).

Subtracting the right tangency equation from the left tangency equation yields \(m_-+H'(\mu_L^{SB}) = m_++H'(\mu_H^{SB})\).

Rearranging, \(H'(\mu_L^{SB})-H'(\mu_H^{SB}) = m_+-m_- = \Delta_W\), which proves~\(H'(\mu_L^{SB})-H'(\mu_H^{SB})=\Delta_W\).

Finally, Bayes plausibility requires \(p^{SB}\mu_H^{SB}+(1-p^{SB})\mu_L^{SB}=\mu_0\).

Solving for \(p^{SB}\) gives \(p^{SB} = \frac{\mu_0-\mu_L^{SB}} {\mu_H^{SB}-\mu_L^{SB}}\).

This completes the proof.
\end{proof}

\begin{proof}[Proof of Corollary~\ref{cor:shannon_explicit}]
Under Shannon entropy, \(H(\mu) = -\mu\ln\mu-(1-\mu)\ln(1-\mu)\).

Differentiating gives \(H'(\mu) = -\ln\mu-1+\ln(1-\mu)+1 = \ln\frac{1-\mu}{\mu}\).

By Proposition~\ref{prop:entropy_tangency}, \(H'(\mu_L^{SB})-H'(\mu_H^{SB}) = \Delta_W\).

Substituting the derivative of Shannon entropy yields \(\ln\frac{1-\mu_L^{SB}}{\mu_L^{SB}} - \ln\frac{1-\mu_H^{SB}}{\mu_H^{SB}} = \Delta_W\).

Using the logarithm rule, \(\ln \left[ \frac{(1-\mu_L^{SB})/\mu_L^{SB}} {(1-\mu_H^{SB})/\mu_H^{SB}} \right] = \Delta_W\).

Equivalently, \(\ln \left[ \frac{(1-\mu_L^{SB})\mu_H^{SB}} {\mu_L^{SB}(1-\mu_H^{SB})} \right] = \Delta_W\).

Exponentiating both sides gives \(\frac{ (1-\mu_L^{SB})\mu_H^{SB} }{ \mu_L^{SB}(1-\mu_H^{SB}) } = e^{\Delta_W}\), which proves~\(\frac{(1-\mu_L^{SB})\mu_H^{SB}}{\mu_L^{SB}(1-\mu_H^{SB})}=e^{\Delta_W}\).

It remains to solve for \(\mu_H^{SB}\) as a function of \(\mu_L^{SB}\). Starting from \(\frac{ (1-\mu_L^{SB})\mu_H^{SB} }{ \mu_L^{SB}(1-\mu_H^{SB}) } = e^{\Delta_W}\), multiply both sides by \(\mu_L^{SB}(1-\mu_H^{SB})\): \((1-\mu_L^{SB})\mu_H^{SB} = e^{\Delta_W}\mu_L^{SB}(1-\mu_H^{SB})\).

Expanding the right-hand side, \((1-\mu_L^{SB})\mu_H^{SB} = e^{\Delta_W}\mu_L^{SB} - e^{\Delta_W}\mu_L^{SB}\mu_H^{SB}\).

Collecting the terms in \(\mu_H^{SB}\) gives \(\mu_H^{SB} \left[ 1-\mu_L^{SB} + e^{\Delta_W}\mu_L^{SB} \right] = e^{\Delta_W}\mu_L^{SB}\).

Therefore \(\mu_H^{SB} = \frac{ e^{\Delta_W}\mu_L^{SB} }{ 1-\mu_L^{SB} + e^{\Delta_W}\mu_L^{SB} }\).

Dividing numerator and denominator by \(e^{\Delta_W}\) yields \(\mu_H^{SB} = \frac{ \mu_L^{SB} }{ \mu_L^{SB} + (1-\mu_L^{SB})e^{-\Delta_W} }\), which proves~\(\mu_H^{SB}=\frac{\mu_L^{SB}}{\mu_L^{SB}+(1-\mu_L^{SB})e^{-\Delta_W}}\).
\end{proof}

\begin{proof}[Proof of Corollary \ref{cor:persuasion_rent_entropy}]
Under Shannon entropy, \(H(\mu) = -\mu\ln\mu-(1-\mu)\ln(1-\mu)\), so \(c(\tau) = H(\mu_0)-\int H(\mu)\,d\tau(\mu)\) is precisely the standard expression for the mutual information \(I(\omega;S)\) generated by the experiment inducing \(\tau\), for any \(\tau\in\mathcal K\) and in particular for any \(\tau\in\mathcal I\); this is an identity about the cost functional \(c\) itself; it holds for every implementable \(\tau\), not only at the second-best optimum. Substituting into \(T(\tau)=\bar U-\int V_M^0\,d\tau+c(\tau)\) from Lemma~\ref{lem:participation_binds} gives \(T(\tau)=\bar U-\int V_M^0\,d\tau+I(\omega;S)\), as claimed; the mediator's own utility remains \(\bar U\) by construction of \(T(\tau)\), so no part of this identity attributes rent to the mediator beyond the outside option.
\end{proof}

\newpage
\bibliographystyle{apalike}
\bibliography{FinalVersion}

\end{document}